\documentclass[letterpaper, 10 pt, conference]{ieeeconf}  

\IEEEoverridecommandlockouts                              
\usepackage{cite}
\usepackage[hidelinks]{hyperref}
\usepackage{amsmath,amssymb,amsfonts}
\usepackage{algpseudocode,algorithm}
\usepackage{graphicx}
\usepackage{dsfont}
\usepackage{textcomp}
\usepackage{subcaption}
\usepackage{mathtools}
\usepackage{xcolor}
\usepackage{tabularx}
\usepackage{academicons}
\usepackage{multirow}
\usepackage{diagbox}
\usepackage[colorinlistoftodos]{todonotes}

\def\Ac{{\mathcal A}}

\def\Lc{{\mathcal L}}

\def\Mc{{\mathcal M}}

\def\Nc{{\mathcal N}}

\def\Rbb{{\mathbb R}}

\def\Sbb{{\mathbb S}}

\def\Vc{{\mathcal V}}

\def\Zbb{{\mathbb Z}}

\def\0{{\bf 0}}

\newcommand{\bitem}{\begin{itemize}}
\newcommand{\eitem}{\end{itemize}}
\newcommand{\btabular}{\begin{tabular}}
\newcommand{\etabular}{\end{tabular}}
\newcommand{\bcenter}{\begin{center}}
\newcommand{\ecenter}{\end{center}}
\newcommand{\bea}{\begin{eqnarray}}
\newcommand{\eea}{\end{eqnarray}}
\newcommand{\bean}{\begin{eqnarray*}}
\newcommand{\eean}{\end{eqnarray*}}
\newcommand{\ba}{\left[ \begin{array}}
\newcommand{\ea}{\\ \end{array} \right]}
\newcommand{\bear}{\begin{array}}
\newcommand{\eear}{\\ \end{array}}

\newcommand{\non}{\nonumber}

\newcommand*{\QEDB}{\hfill\ensuremath{\blacksquare}}%
\newcommand*{\QET}{\hfill\ensuremath{\triangleleft}}
\newcommand{\norm}[1]{\left\lVert#1\right\rVert}

\newcounter{subequation}
\def\beasub{\addtocounter{equation}{+1}
\setcounter{subequation}{\value{equation}}
\setcounter{equation}{0}
\renewcommand{\theequation}{\arabic{subequation}\alph{equation}}
\begin{eqnarray}}
\def\eeasub{\end{eqnarray}
\setcounter{equation}{\value{subequation}}
\renewcommand{\theequation}{\arabic{equation}}}

\def\inf{\operatornamewithlimits{inf\vphantom{p}}}

\newtheorem{Lemma}{Lemma}
\newtheorem{Theorem}{Theorem}
\newtheorem{Definition}{Definition}
\newtheorem{Assumption}{Assumption}
\newtheorem{Remark}{Remark}

\newtheorem{Proposition}{Proposition}

\title{\LARGE \bf
Bilateral Cognitive Security Games in Networked Control Systems under Stealthy Injection Attacks
}

\author{Anh Tung Nguyen, Quanyan Zhu, and Andr{\'e} Teixeira
\thanks{This work is supported by the Swedish Research Council under the grant 2021-06316 and by the Swedish Foundation for Strategic Research.}
\thanks{Anh Tung Nguyen and Andr{\'e} Teixeira are with the Department of Information Technology, Uppsala University, PO Box 337, SE-75105, Uppsala, Sweden
{\tt\small \{anh.tung.nguyen, andre.teixeira\}@it.uu.se}.}
\thanks{Quanyan Zhu is with the Department of Electrical and Computer Engineering, New York University, Brooklyn, NY, 11201, USA {\tt\small qz494@nyu.edu.}}
}

\begin{document}

\maketitle
\thispagestyle{empty}
\pagestyle{empty}

\begin{abstract}
This paper studies a strategic security problem in networked control systems under stealthy false data injection attacks. 
The security problem is modeled as a bilateral cognitive security game between a defender and an adversary, each possessing cognitive reasoning abilities. The adversary with an adversarial cognitive ability strategically attacks some interconnections of the system with the aim of disrupting the network performance while remaining stealthy to the defender. Meanwhile, the defender with a defense cognitive ability strategically monitors some nodes to impose the stealthiness constraint with the purpose of minimizing the worst-case disruption caused by the adversary. Within the proposed bilateral cognitive security framework, the preferred cognitive levels of the two strategic agents are formulated in terms of two newly proposed concepts, cognitive mismatch and cognitive resonance. Moreover,
we propose a method to compute the policies for the defender and the adversary with arbitrary cognitive abilities.  A sufficient condition is established under which an increase in cognitive levels does not alter the policies for the defender and the adversary, ensuring convergence.  The obtained results are validated through numerical simulations.

\end{abstract}

\section{INTRODUCTION}
Networked control systems (NCSs) constitute the backbone of critical infrastructure, including power grids, transportation networks, and water distribution systems \cite{molzahn2017survey,polycarpou2023smart,conti2021survey}. The integration of such systems with open communication platforms such as public internet protocols and wireless technologies, has, however, introduced significant vulnerabilities to cyber-physical attacks. The consequences of successful cyber intrusions in these domains are often severe, posing risks not only to economic stability but also to public safety and societal well-being. Historical incidents, such as the Stuxnet malware targeting Iranian industrial control systems in 2010 \cite{falliere2011w32} and the Industroyer cyberattack on Ukraine's power grid in 2016 \cite{kshetri2017hacking}, underscore the capacity of advanced cyber threats to inflict substantial damage. These events have heightened global awareness regarding the security of NCSs, rendering cybersecurity a critical and urgent research priority within the control systems community.

While traditional models often frame cyber-physical security as an interaction between fully rational agents within Stackelberg or Nash equilibria frameworks \cite{li2018false,yuan2019stackelberg,shukla2022robust,umsonst2021bayesian}, real-world defenders and adversaries are typically human operators or human-influenced entities, who exhibit bounded rationality and finite-depth reasoning \cite{kanellopoulos2019non,li2019decision}. This gap between theory and practice motivates the need for security models that more realistically capture the cognitive limitations \cite{camerer2004cognitive} and decision-making heuristics employed by both defenders and attackers.
In light of these challenges, this work advocates for a bilateral cognitive security game formulation \cite{huang2023cognitive}, where both players operate under finite cognitive hierarchies, resulting in a more nuanced and applicable framework for securing NCSs under stealthy false data injection attacks.

In this paper, we investigate the security of a continuous-time NCS represented by a graph, where each node corresponds to a one-dimensional subsystem. The system is subject to stealthy false data injection attacks that aim to compromise the integrity of inter-node communications. Specifically, the adversary strategically selects a subset of nodes from which to launch stealthy attacks, thereby maximizing the degradation of the system's overall performance. In response, a defender deploys a limited number of sensors to monitor selected node outputs, imposing stealth constraints that restrict the adversary’s ability to act without detection. The interplay between these two opposing agents gives rise to a security game over the network, as depicted in Fig.~\ref{fig:problem}. Within this framework, we study the system’s vulnerability and provide the following contributions:
\begin{enumerate}
    \item We propose a novel game-theoretic framework that incorporates the concepts from cognitive hierarchy models \cite{camerer2004cognitive} and Stackelberg prediction games \cite{bruckner2011stackelberg} 
    to model security interactions between an adversary and a defender operating with finite-depth reasoning. 
    \item A semidefinite programming (SDP) is developed to compute the policy and the maximum disruption for the adversary with finite-depth reasoning. 
    \item Given that the defender assumes the adversary operates at lower reasoning levels, we develop a mixed-integer SDP formulation to find the defender’s optimal monitoring strategy in response to the adversary’s attack policy.
\end{enumerate}

\textit{Notation:}
$\Rbb^n \, (\Rbb^n_{>0},\,\Rbb^n_{\geq 0})$ and $\Rbb^{n \times m}$ stand for sets of real (positive, non-negative) $n$-dimensional vectors and real $n$-by-$m$ matrices, respectively; $\Zbb_{\geq 0}$ denote a set of non-negative integers;
the set of $n$-by-$n$ symmetric matrices is denoted as $\Sbb^n$.
We denote $A \preceq 0$ if $-A$ is a positive semi-definite matrix.
An $i$-th column of the $n$-by-$n$ identity matrix is denoted as $e_i$.
The space of square-integrable  functions is defined as $\Lc_{2} \triangleq \bigl\{f: \Rbb_{>0} \rightarrow \Rbb^n ~|~ \norm{f}^2_{\Lc_2 [0,\infty]} < \infty \bigr\}$ and the extended space be defined as $\Lc_{2e} \triangleq \bigl\{ f: \Rbb_{>0} \rightarrow \Rbb^n ~|~ \norm{f}^2_{\Lc_2 [0,H]} < \infty,~ \forall~ 0 < H < \infty \bigr\} $ where $\norm{f}_{\Lc_2 [0,H]}^2 \triangleq \int_{0}^{H} \norm{f(t)}_2^2 \, \text{d}t$.
The notation $\norm{f}^2_{\Lc_2}$  is used  as shorthand for the norm $\norm{f}_{\Lc_2 [0,H]}^2$ if the time horizon $[0,H]$ is clear from the context. We denote $\circ$ as an element-wise multiplication operator and $\textbf{1}$ as an all-one vector with an appropriate dimension. 
\begin{figure}[!t]
    \centering
    \includegraphics[width=0.9\linewidth]{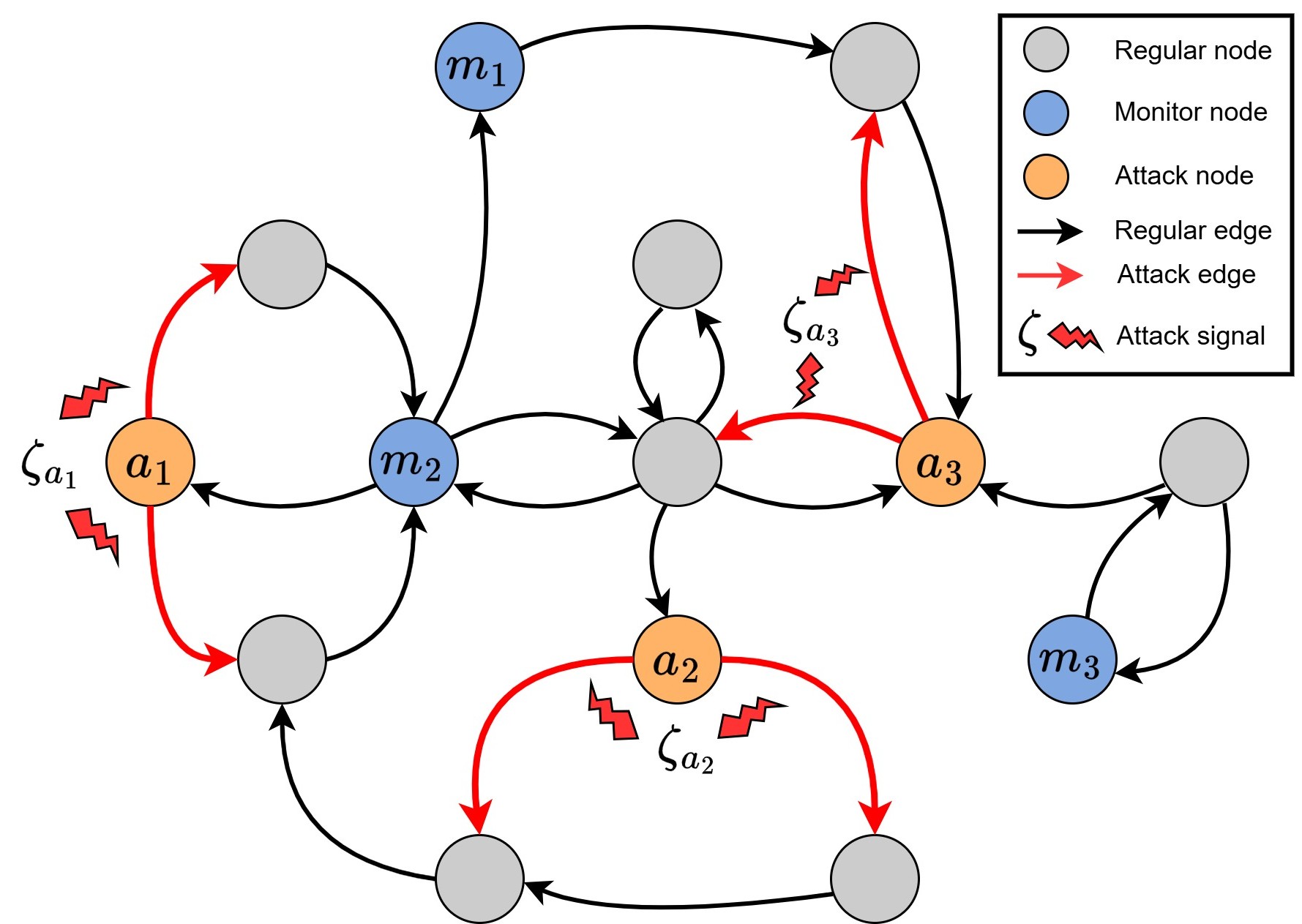}
    \caption{A networked control system under stealthy injection attacks. An adversary injects attack signals into the information sent from orange nodes to their neighbors while a defender monitors the outputs of blue nodes.}
    \label{fig:problem}
\end{figure}
\section{PROBLEM DESCRIPTION}
\label{sec:problem_formulation}
In this section, we first describe an NCS, with a global performance metric, under stealthy attacks. 
Subsequently, a competition between an adversary and a defender with finite-depth reasoning levels is briefly described to show the main problem we study throughout the paper.
\subsection{Networked control systems under attacks}
This subsection introduces an NCS in normal and attacked operations and the resources of the adversary and the defender. Subsequently, we formulate the performance loss as a quantification of malicious impacts on the system.  
\subsubsection{Networked control systems without attacks}
We consider the NCS depicted in Fig.~\ref{fig:problem} where each node has a one-dimensional dynamics for ease of exposition. The dynamics of each node can be further extended to be heterogeneous with different dimensions, which will be left for future extension. 
The dynamics of each node is described in the following form:
\begin{align}
	\dot x_i(t) &= A_{ii} x_i(t) + u_i(t), 
 ~ i \in \Vc \triangleq \bigl\{1,\,2,\ldots,\,N\bigr\},
	\label{sys:xi}
\end{align}
where $x_i(t) \in \Rbb$ and $u_i(t) \in \Rbb$ are the state and the local control input, respectively. 
The local parameter $A_{ii}$ is given. 
Each node $i \in \Vc$ is controlled by the following control law:
\begin{align}
	u_i(t) = -(\theta_i + A_{ii}) x_i(t) +  \sum_{j \, \in \, \Nc_i} A_{ij} \big(x_j(t) - x_i(t)\big),
	\label{sys:ui}
\end{align}
where $\theta_i \in \Rbb_{>0}$ is the control gain at node $i$ and $A_{ij} \in \Rbb_{>0}$ is the interaction gain of node $i$ with its neighbors. The set of neighboring nodes of node $i$ is denoted as $\Nc_i$. 

\subsubsection{Defense resources}
To detect potential malicious activities, the defender selects a subset of the node set $\Vc$ as a set of monitor nodes, denoted as $\Mc = \{m_1,m_2,\ldots,m_{| \Mc |} \} \subset \Vc$. A sensor is placed on each monitor node to monitor its state, where the number of utilized sensors should be constrained for practical reasons.
Let us denote $\beta \in \Rbb_{>0}$ as the sensor budget that is the maximum number of utilized sensors, i.e., 
\begin{align}
    |\Mc| \leq \beta. \label{sensor_budget}
\end{align} 
More specifically, the defender monitors the following output measurements:
\begin{align}
    y_m(t) = e^\top_m x(t),~~ \forall \, m \in \Mc. \label{sys:ym}
\end{align}
At each monitor node $m \in \Mc$, a corresponding alarm threshold $\delta_m \in \Rbb_{>0}$ is assigned. The presence of the adversary is detected if the output energy for a given time horizon $[0,H]$ of at least one monitor node crosses its corresponding alarm threshold, i.e.,
\begin{align}
    \norm{y_m}^2_{\Lc_2[0,H]} > \delta_m. \label{alarm_deltam}
\end{align}
Further, each node $i$ has a cost $\kappa_i \in \Rbb_{>0}$, resulting in the following sensor cost of the monitor set $\Mc$:
\begin{align}
    c_s(\Mc) = \sum_{m \, \in \, \Mc} \kappa^\top e_m, \label{cost_sensor}
\end{align}
where $\kappa = \left[ \kappa_1, \kappa_2, \ldots, \kappa_N \right]^\top$ is a vector of all sensor costs.

\subsubsection{Adversary resources}
The adversary selects exactly $\alpha~(\alpha \leq N)$ nodes on which to conduct false data injection attacks on the information sent from these $\alpha$ attack nodes to their neighbors (the orange nodes in Fig.~\ref{fig:problem}). More specifically, these $\alpha$ nodes are not directly affected by attacks, but their neighboring nodes are. Henceforth, these $\alpha$ nodes are called attack nodes.

Let us denote a set of $\alpha$ attack nodes as follows:
\begin{align}
    \Ac \triangleq \{ a_1, a_2, \ldots, a_{\alpha} \} \subset \Vc. \label{attack_set}
\end{align}
For each attack node $a_i \in \Ac$, the adversary designs an additive attack signal $\zeta_{a_i}(t)$ into the information sent from the attack node $a_i$ to all its neighbors, which is assumed to have bounded energy: 
\begin{align}
   \norm{\zeta_{a_i}}^2_{\Lc_2[0,\infty]}  \leq E < \infty, ~\forall \, a_i \in \Ac, \label{zeta_bounded}
\end{align}
where the maximum attack energy $E \in \Rbb_{>0}$ is given.
As a consequence, the control input in \eqref{sys:ui} under false data injection attacks can be represented as follows:
\begin{align}
    u^a_i(t) \triangleq u_i(t) + \begin{cases}
       A_{i \, a_i} \zeta_{a_i}(t),~\text{if}~ i \in  \Nc_{a_i} \,\&\, a_i \in \Ac, \\
        0, ~~ \text{otherwise}, 
    \end{cases} \label{sys:uia}
\end{align}
where $u_i(t)$ is given in \eqref{sys:ui} and the superscript ``$a$'' stands for signals subjected to attacks.

\subsubsection{Networked control systems under stealthy attacks}
Given the attack set $\Ac$ in \eqref{attack_set}, let us denote the corresponding attack input matrix as $B_\Ac \triangleq [Ae_{a_1},Ae_{a_2},\ldots,Ae_{a_{\alpha}}]$ and the attack signal vector as $\zeta(t) \triangleq [\zeta_{a_1}(t),\zeta_{a_2}(t),\ldots,\zeta_{a_{\alpha}}(t)]^\top$, where $A$ is defined as follows:
\begin{align*}
    A \triangleq \ba{cccc}
    0 & A_{12} & \ldots & A_{1N} \\
    A_{21} & 0 & \ldots & A_{2N} \\
    \vdots & \ldots & \ddots & \vdots \\
    A_{N1} & A_{N2} & \ldots & 0
    \ea.
\end{align*}
Note that $A_{ij} > 0$ if there is a connection between nodes $i$ and $j$ and $A_{ij} = 0$ otherwise.
For convenience, let us denote $x(t)$ as the state of the entire network, $x(t) \triangleq \big[x_1(t),~x_2(t),\ldots,~x_N(t)\big]^\top$. Therefore, the system \eqref{sys:xi} under the attacked control input \eqref{sys:uia} results the following system under attacks:
\begin{align}
    \dot x^a(t) &= A_c x^a(t)  + B_\Ac  \zeta(t), \label{sys:xa} \\
    p^a(t) &= W x^a(t), \label{sys:pa} \\
    y^a_m(t) &= e_m^\top x^a(t), ~ \forall m \in \Mc, \label{sys:yma}
\end{align}
where $p^a(t)$ is the performance output of the NCS and $y^a_m(t)$ is the monitor output under attacks at node $m$. The performance weighting factor $W \in \Rbb^{N \times N}$ is given and
$A_c \triangleq -\text{diag}([\theta_1, \theta_2, \ldots, \theta_N]) - A \textbf{1} + A$.

It is worth noting that the matrix $A_c$ in \eqref{sys:xa} alternatively represents the in-degree Laplacian matrix of a graph with self-loops \cite{west2001introduction}. The self-loops can be chosen such that $A_c$ is Hurwitz. As a result, the system can be assumed to converge to its equilibrium before being exposed to attacks, enabling us to use the following assumption.
\begin{Assumption}
    \label{assumption:x0}
    The system \eqref{sys:xa} is at its equilibrium $x_e = 0$ before being affected by the attack signal $\zeta(t)$.  \QET
\end{Assumption}

On the other hand, the energy constraint \eqref{zeta_bounded} allows us to assume that $x^a(\infty) = 0$ (see more detail in our previous work \cite[Sec. II]{nguyen2024Submit}).

Similar to robust control, the performance of the NCS for a given, possibly infinite, time horizon $[0,H]$ is formulated as follows:
\begin{align}
    J^a \triangleq \norm{p^a}_{\Lc_2[0,H]}^2.
    \label{sys:J}
\end{align}
From the adversary perspective, we assume that the purpose of the adversary is to maximally disrupt the performance \eqref{sys:J} subject to the model \eqref{sys:xa}-\eqref{sys:yma} while remaining stealthy to the defender {(see the discussion on the importance of the stealthiness in \cite[Sec. II.E]{umsonst2021bayesian}).
This adversarial purpose allows us to mainly focus on stealthy injection attacks which will be defined in the following. 
\begin{Definition}
    [Stealthy injection attacks] \label{def:stealthy_attacks}
    Consider the system
    \eqref{sys:xa}-\eqref{sys:yma} with monitor outputs $y_m^a(t) = e_m^\top x^a(t)$ for every $m \in \Mc$, which is a set of monitor nodes. The attack $\zeta(t)$ on the system \eqref{sys:xa}-\eqref{sys:yma} is defined as a stealthy injection attack if the following condition $\norm{y_m^a}^2_{\Lc_2} \leq \delta_m$ holds for all $m \in \Mc$. \QET
\end{Definition} 

The worst-case impact of stealthy injection attacks for an infinite time horizon, which is referred to as the worst-case disruption henceforth, is formulated as follows:
\begin{align}
    Q(\Mc, \Ac) \triangleq 
    ~&\sup_{ \{e_j^\top \zeta\}_{\forall j \in \{1,2,\ldots,\alpha\}} \in \Lc_{2e} }~
    \norm{p^a}_{\Lc_2}^2 
    \label{Q_sup} \\
    \text{s.t.}~&
    \eqref{sys:xa}-\eqref{sys:yma}, \, x^a(0) = 0, \, x^a(\infty) = 0, \non \\
    &\norm{y_m^a}^2_{\Lc_2} \leq \delta_m,~ \forall \, m \in \Mc, \non \\
    &\norm{e_j^\top \zeta}^2_{\Lc_2} \leq E,~ \forall \, j  \in \{1,2,\ldots, \alpha \}. \non    
\end{align}
The worst-case disruption \eqref{Q_sup} is also
called an Attack-Energy-Constrained Output-to-Output gain security metric \cite{gallo2025switching} for a given set of attack nodes $\Ac$ and a given set of monitor nodes $\Mc$. By observing \eqref{Q_sup}, we introduce a security problem considering the defender and the adversary as humans in the following.

\subsection{Bilateral Cognitive Security} 
From the game theory perspective \cite{bacsar1998dynamic}, the worst-case disruption \eqref{Q_sup} can be seen as a game payoff for the two strategic players, i.e., the defender chooses $\Mc$ to minimize \eqref{Q_sup} while the adversary selects $\Ac$ to maximize \eqref{Q_sup}. For a fixed $\Ac$, we can argue that the defender can outsmart the adversary by deviating the selection of $\Mc$ to obtain a smaller worst-case disruption. The same argument for the adversary remains true for a fixed $\Mc$.

However, in practical security settings, strategic agents do not necessarily exhibit full rationality or perfect strategic foresight. Instead, they, as humans, make decisions with finite-depth cognitive reasoning about their opponent’s thought process. This introduces a bilateral cognitive security problem, where the defender and the adversary operate under asymmetric cognitive hierarchies, each making strategic choices based on limited beliefs about the opponent’s reasoning level. To model these interactions, we introduce the Cognitive Hierarchy-$k$ (CH-$k$) framework and the Stackelberg prediction game, capturing how both players respond to their opponent’s strategic depth reasoning in the next section.

\section{BILATERAL COGNITIVE SECURITY GAME}
We introduce a mixture of the cognitive hierarchy model \cite{camerer2004cognitive} and the Stackelberg prediction game in the Machine Learning community \cite{bruckner2011stackelberg} where the defender acts as a leader and the adversary acts as a follower. 
The cognitive hierarchy model assumes that the strategic players base their actions on a finite depth of reasoning, which is called Cognitive Hierarchy, about the likely actions of the other players. On the other hand,
the Stackelberg prediction game deviates from the classical Stackelberg game \cite{bacsar1998dynamic} by relaxing the assumption that the follower certainly observes the leader's decision. Inspired by \cite{bruckner2011stackelberg}, we assume that the follower can compute the best response of the leader rather than directly observing the leader's action. 
Next, we describe how the players with different cognitive hierarchies interact with their opponents in the following definition.

\begin{Definition}
    [Asymmetric CH-\(k\) reasoning] \label{def:asym_chk}
    Let \( k \in \mathbb{Z}_{\geq 0} \) be a non-negative integer. A player (defender or adversary) is said to employ \textit{CH-\(k\) reasoning} if 
    \begin{enumerate}
        \item the defender chooses their strategy responding to strategies of the adversary with cognitive levels strictly less than \(k\).
        \item the adversary chooses their strategy responding to strategies of the defender with the same cognitive level \(k\). \QET
    \end{enumerate}
\end{Definition}

To facilitate the use of Definition~\ref{def:asym_chk},
let us denote the actions chosen by the defender and the adversary with CH-$k$ as $\Mc_k$ and $\Ac_k$ where $k \geq 0$, respectively.
We begin with the defender with zero cognitive level in the following.
\paragraph{CH-$0$ defender}
Followed by the literature on the cognitive hierarchy model \cite{camerer2004cognitive}, we assume that the CH-$0$ defender chooses the monitoring policy $\Mc^0$ randomly regardless of the existence of the adversary. 
\paragraph{CH-$k$ adversary ($k \geq 0$)}
Inspired by the Stackelberg prediction game framework \cite{bruckner2011stackelberg},
the CH-$k$ adversary ($k \geq 0$) finds the best response $\Ac_k$ against the CH-$k$ defender choosing $\Mc_k$. Therefore, the adversary solves the following optimization problem:
\begin{align}
    \Ac_k \triangleq \underset{\Ac \in \Vc, |\Ac| = \alpha}{\arg\max}Q(\Mc_k, \Ac), \label{def_Ak}
\end{align}
where $Q(\Mc_k, \Ac)$ is the worst-case disruption for a given pair of $(\Mc_k, \Ac)$ defined in \eqref{Q_sup}.

\paragraph{CH-$k$ defender ($k \geq 1$)}
The CH-$k$ defender chooses monitor set $\Mc_k$ while considering the adversary has lower cognitive hierarchies, i.e., the CH-$i$ adversary for all $i \in [0, k-1]$.
Consequently, the CH-$k$ defender finds their best response by solving the following optimization problem: 
\begin{align}
    \Mc_k \triangleq \underset{\Mc \in \Vc, |\Mc| \leq \beta}{\arg \min} 
    R(\Mc, \{\Ac_i\}_{\forall i \in [0, k-1]}).
     \label{def_Mk}
\end{align}
Here
\begin{align*}
    R(\Mc, \{\Ac_i\}_{\forall i \in [0, k-1]}) \triangleq \bigg[  c_s(\Mc) + \max_{i \in [0, k-1]} Q(\Mc, \Ac_i) \bigg],
\end{align*}
where $c_s(\Mc)$ is the cost of sensors defined in \eqref{cost_sensor} and the second term is the worst-case disruption against the adversary with different CHs.
Note that this rationality formulation can contain the conservative case where only CH-($k$-1) adversary is considered which is referred to as level-k reasoning players in \cite{kanellopoulos2019non}. 
For a later use, let us represent the policy $\Mc_k$ of the CH-$k$ defender as a binary variable $z_k \in \{0,1\}^N$ such that the following equality holds true:
\begin{align}
    z_k \triangleq \sum_{m \in \Mc_k} e_{m}. \label{def_zk}
\end{align}

Based on the strategies chosen by the players (defender or adversary) with finite-depth reasoning, we introduce the reasoning outcome in the following.
\begin{Definition}
    [CH-\((k,\ell)\) reasoning outcome] \label{def:reasoning_outcome}
    Let \( k, \ell \in \mathbb{Z}_{\geq 0} \) denote the cognitive reasoning levels of the defender and adversary, respectively. A \textit{CH-\((k,\ell)\) reasoning outcome} is defined as the following tuple
    \begin{align}
        (\Mc_k, \Ac_\ell, Q(\Mc_k, \Ac_\ell)), \label{reasoning_outcome}
    \end{align}
    where:
    \begin{itemize}
    \item \( \Mc_k \) is the monitoring strategy selected by a CH-\(k\) defender,
    \item \( \Ac_\ell \) is the attack strategy selected by a CH-\(\ell\) adversary,
    \item \( Q(\Mc_k, \Ac_\ell) \) is the resulting worst-case disruption. \QET
    \end{itemize}
\end{Definition}

The reasoning outcome defined in \eqref{reasoning_outcome} captures the joint strategic consequence of asymmetric cognitive reasoning and serves as a metric for evaluating the system's resilience against mismatched strategic depth reasoning between the defender and the adversary. The possibly mismatched reasoning between players is categorized in the following two definitions.
\begin{Definition}
    [Cognitive mismatch] \label{def:cognitive_mismatch}
    Let the defender employ CH-\(k\) reasoning and the adversary employ CH-\(\ell\) reasoning. A \textit{cognitive mismatch} occurs when the reasoning model assumed by one player about the other’s behavior deviates from the actual reasoning level employed by that player. \QET
\end{Definition}
\begin{Definition}
    [Cognitive resonance] \label{def:cognitive_resonance}
    Let the defender employ CH-\(k\) reasoning and the adversary employ CH-\(\ell\) reasoning. A \textit{cognitive resonance} occurs if the reasoning models used by both players accurately reflect the actual strategies and reasoning levels of their opponents. \QET
\end{Definition}

The following propositions suggest how deep cognitive reasoning the defender and the adversary should employ to benefit the reasoning outcome.
\begin{Proposition}
   \label{prop:adv_mismatch} 
    Let \( \Mc_k \) be a fixed CH-\(k\) monitoring policy. Suppose the attacker at CH-\(\ell\) chooses their policy defined in \eqref{def_Ak}. 
    The CH-\((k,\ell)\) reasoning outcome \( (\Mc_k, \Ac_\ell, Q(\Mc_k, \Ac_\ell)) \) gains no benefit for the adversary if there is a
    \textit{cognitive mismatch} $\ell \neq k$ , i.e.,
    \begin{align}
        Q(\Mc_k, \Ac_\ell) \leq Q(\Mc_k, \Ac_k) ~~\forall \ell \neq k. \tag*{\QET}
    \end{align}
\end{Proposition}
\begin{proof}
    By definition, $\Ac_k$ is a maximizer of the optimization problem \eqref{def_Ak}. As a result, another policy $\Ac_\ell$ deviated from $\Ac_k$ with the same cardinality, i.e., $|\Ac_\ell| = |\Ac_k|$, does not gain more worst-case disruption. 
\end{proof}

\begin{Proposition}
    \label{prop:def_mismatch}
    Suppose that the cost \eqref{cost_sensor} is the same for all the sensors.
    Let $\Ac_\ell$ be a fixed CH-$\ell$ adversary policy. Suppose the defender at CH-$k$ chooses their policy defined in \eqref{def_Mk}. The CH-\((k,\ell)\) reasoning outcome \( (\Mc_k, \Ac_\ell, Q(\Mc_k, \Ac_\ell)) \) has a non-negative benefit for the defender if the defender has a one-level higher cognitive reasoning level than the adversary $k = \ell + 1$, i.e.,
    \begin{align}
        Q(\Mc_{\ell+1}, \Ac_\ell) \leq Q(\Mc_\ell, \Ac_\ell). \tag*{\QET}
    \end{align}
\end{Proposition}
\begin{proof}
    From \eqref{def_Mk} and $c_s(\Mc_\ell) = c_s(\Mc_{\ell+1})$, one obtains $\max_{i\in[0,\ell]} Q(\Mc_{\ell+1},\Ac_i) \leq \max_{i \in [0,\ell]} Q(\Mc_{\ell}, \Ac_i)$. On the other hand, $\max_{i \in [0,\ell]} Q(\Mc_{\ell}, \Ac_i) = Q(\Mc_{\ell}, \Ac_{\ell})$ from Proposition~\ref{prop:adv_mismatch} while by definition $Q(\Mc_{\ell+1},\Ac_\ell) \leq \max_{i\in[0,\ell]} Q(\Mc_{\ell+1},\Ac_i)$. Consequently, one obtains $Q(\Mc_{\ell+1},\Ac_\ell) \leq Q(\Mc_{\ell}, \Ac_{\ell})$, concluding the proof.
\end{proof}

The result presented in Proposition~\ref{prop:adv_mismatch} suggests the adversary to employ the cognitive resonance (see Definition~\ref{def:cognitive_resonance}). On the other hand, the result of Proposition~\ref{prop:def_mismatch} encourages the defender to have a cognitive mismatch (see Definition~\ref{def:cognitive_mismatch}) to gain more benefits rather than having a cognitive resonance.

\begin{Remark}
    In risk management, the defender must account for the worst-case scenario, where their defense policy may be exposed to the adversary. Moreover, the defender lacks precise knowledge of the adversary’s cognitive reasoning depth. The formulation in \eqref{def_Mk} addresses these uncertainties by ensuring robustness against different adversary strategies. On the other hand, adversaries often gather information about system dynamics and defensive measures before making their decisions. The formulation in \eqref{def_Ak} captures the worst-case scenario where the adversary has perfect information, providing the defender with valuable insights for refining their strategy. This framework allows the defender to enhance security measures by considering high-cognitive-reasoning adversaries and improving resilience against informed attacks.  \QET
\end{Remark}

In the following section, we show how the CH-$k$ policies for the two strategic agents are computed.

\section{COGNITIVE HIERARCHY POLICY COMPUTATION}
This section provides a method to compute CH-$k$ policies of the adversary \eqref{def_Ak} and the defender \eqref{def_Mk}. Then, an algorithm is provided to show a procedure for computing a policy for the defender with an arbitrary cognitive hierarchy.

\subsection{CH-$k$ adversary policy}
First, the worst-case disruption \eqref{Q_sup} considers the impact of finite energy attack signals on a stable system, it is always bounded and tractable for any pair of $\Ac$ and $\Mc$ \cite[Lemma 2]{nguyen2024Submit}. Next, for a fixed CH-$k$ defense policy $\Mc_k$, the maximum worst-case disruption \eqref{def_Ak} is computed by the following lemma, which is improved from \cite[Lemma 3]{nguyen2024Submit}. 
\begin{Lemma}
    \label{lem:adversary_chk}
    Suppose the CH-$k$ defender chooses $\Mc_k$ denoted as $z_{k}$ in \eqref{def_zk}. For each admissible attack set $\Ac \subset \Vc~(|\Ac| = \alpha)$, 
    a tuple of variables $( \gamma_\Ac,\,\psi_\Ac,P_\Ac, \epsilon_\Ac) \, \in \, \Rbb^N_{>0} \times \Rbb_{>0}^{\alpha_k} \times \Sbb^N \times \Rbb_{\geq 0}$ is defined correspondingly. 
    The maximum worst-case disruption \eqref{def_Ak} is denoted as $Q_k$, which is the optimal solution to the following SDP:
    \begin{align}
    &\min_{Q_k,\, \{
    \gamma_\Ac,\,\psi_\Ac,\,P_\Ac, \epsilon_\Ac \}_{\forall  \Ac} }~~ r Q_k - \sum_{\forall \Ac} \epsilon_\Ac \label{Q_sdp}  \\
    \text{s.t.}
    &~~~~~
    Q_k \in \Rbb_{>0}, ~
    \gamma_\Ac \in \Rbb^N_{>0},~ \psi_\Ac \in \Rbb_{>0}^{\alpha}, \non \\
    &~~~~~
    P_\Ac \in \Sbb^{N}, ~ \epsilon_\Ac \in \Rbb_{\geq 0}, \non \\
    &~~~~~  
    \delta^\top \gamma_\Ac
    + E \, \textbf{1}^\top  \psi_\Ac + \epsilon_\Ac  \leq Q_k, \non \\
    &~~~~~
    \ba{cc}
    A_c^\top P_\Ac + P_\Ac A_c + W^2 & P_\Ac B_\Ac \\
    B_\Ac^\top P_\Ac & -\, \textbf{diag}(\psi_\Ac)
     \ea 
    \non \\
    &~~~~~
    -  \textbf{diag} \Bigg(\ba{c} 
    \gamma_\Ac \circ z_{k} \\ 0 \ea \Bigg) \preceq 0,
    ~~\forall \Ac. \non 
    \end{align}
    Here, $r$ is a given large positive number.
    Further, the CH-$k$ adversary policy $\Ac_k$ is found such that its corresponding solution $\epsilon_{\Ac_k}$ equals to zero, i.e., $\epsilon_{\Ac_k} = 0$.
    \QET
\end{Lemma}
\begin{proof}
    See Appendix~\ref{app:pflem_adversary_chk}. 
\end{proof}

The result of Lemma~\ref{lem:adversary_chk} provides us a method to compute the maximum disruption and the best response of the CH-$k$ adversary policy $\Ac_k$ for a given monitor set $\Mc_k$ chosen by the CH-$k$ defender. 

\begin{Remark}
    Note that finding the best response for the adversary falls outside the scope of our previous work \cite{nguyen2024Submit}. 
    Lemma~\ref{lem:adversary_chk} improves the result reported in \cite[Lemma 3]{nguyen2024Submit} by introducing a new variable $\epsilon_\Ac$, which enables us to find the best response for the CH-$k$ adversary policy. Moreover, since \eqref{Q_sdp} is only interested in $\epsilon_{\Ac_k} = 0$, we can add a constraint $\epsilon_\Ac \leq \bar \epsilon$ where $\bar \epsilon$ is a very small positive scalar. This addition could remove unnecessary computations for other variables $\epsilon_{\Ac}~(\Ac \neq \Ac_k)$ in practice. 
    \QET
\end{Remark}

\begin{algorithm}[!t]
    \caption{CH-$q$ policies} \label{al:chq_policies}
    \begin{algorithmic}[1]
   \Statex{{\bf Output:}} The policies for the CH-$q$ defender and the CH-$q$ adversary ($q \in \Zbb_{\geq 0}$).
   \Statex{{\bf Input:}} System matrix $A_c$, performance weighting factor $W$, alarm threshold $\delta$, cost of sensors $\kappa$, maximum attack energy $E$.
   \Statex{{\bf Initialization:}} CH-$0$ defense policy $\Mc^0$ is chosen randomly and $k = 0$.
   \State Solve \eqref{Q_sdp}.
   \State Find $\Ac_k$ for the CH-$k$ adversary by checking $\epsilon_{\Ac_k} = 0$. 
   \State $k = k + 1$.
   \If {$k \leq q$} 
   \State Update CH-$i$ adversary policy $\{\Ac_i\}_{\forall i \in [0,k-1]}$.
   \State Solve \eqref{R_sdp} to obtain $z_{k}$.
   \State Extract $\Mc_k$ from $z_{k}$ via \eqref{def_zk}. 
   \State \textbf{Back} to Step 1.
   \Else 
   \State \textbf{Return} $\Ac_k$ for the CH-$q$ adversary,
   \State \textbf{Return} $\Mc_k$ for the CH-$q$ defender.
   \EndIf
   \end{algorithmic}
\end{algorithm}

\subsection{CH-$k$ defender policy}
For a given CH-$i$ adversary policy $\Ac_i~(i \in [0, k-1])$, the following theorem presents how the CH-$k$ defender finds their policy \eqref{def_Mk}.
\begin{Theorem}
    \label{thm:chk_defender}
    For each CH-$i$ adversary choosing the policy $\Ac_i$, a tuple of variables $(\omega_{i},\psi_{i},P_{i}) \, \in \, \Rbb^N_{\geq 0} \times \Rbb_{>0}^{\alpha} \times \Sbb^N$ is defined correspondingly. 
    Recall the CH-$k$ defense policy $\Mc_k$ is denoted as a binary variable $z_k \in \{0,1\}^N$ in \eqref{def_zk}, which is the optimal solution to the following mixed-integer SDP problem:
    \begin{align}
    &\min_{z_k, R_k, \{\omega_{i},\, \psi_{i},\, P_{i} \}_{\forall i \in [0, k-1]}  }  ~~ \kappa^\top z_k 
    + Q_k    \label{R_sdp} \\
    &\text{s.t.}~
    z_k \in \{0,1\}^N, \,  \omega_{i} \in \Rbb_{\geq 0}^N, \, \psi_{i} \in \Rbb^{\alpha}_{>0}, 
    P_{i} \in \Sbb^{N},  \non \\
    &~~~~~
    Q_k \in \Rbb_{> 0}, ~\textbf{1}^\top z \leq \beta, \, \omega_{i} \, \leq \,  M_\infty z_k, \, 
    \non \\ 
    &~~~~~ \delta^\top \omega_{i} + E \textbf{1}^\top \psi_{i} \leq Q_k,
    \non \\
    &~~~~~
    \ba{cc}
    A_c^\top P_{i} + P_{i} A_c + W^2 & P_{i} B_{\Ac_i} \\
    B_{\Ac_i}^\top P_{i} & -\, \textbf{diag}(\psi_{i})
    \ea \non \\
    &~~~~
    - \textbf{diag} \Bigg(\ba{c} 
    \omega_{i} \\ 0 \ea \Bigg) \preceq 0, 
    ~\forall i \in [0, k-1], \non
    \end{align}
    where 
    $\kappa = [\kappa_1,\kappa_2,\ldots,\kappa_N]^\top \in \Rbb^N_{>0}$ is a given cost vector of sensors, $\delta = \big[ \delta_1, \delta_2, \ldots, \delta_N \big]^\top \in \Rbb^N_{>0}$
    is a given alarm threshold vector of all the nodes,
    $M_\infty$ is a given large positive number, also called a ``big M'' \cite{milovsevic2023strategic},
    $\textbf{1}$ stands for an all-one vector with a proper dimension, and $B_{\Ac_i}$ corresponds to the CH-$i$ adversary policy $\Ac_i$. 
    \QET
\end{Theorem}

\begin{proof}
    See Appendix~\ref{sec:pfthm_defender_chk}.
\end{proof}

The result of Theorem~\ref{thm:chk_defender} enables us to find the optimal policy for the CH-$k$ defender by solving the mixed-integer SDP problem \eqref{R_sdp}, which can be done using the latest version of YALMIP \cite{lofberg2004yalmip}. 
Let us assume that the defense desires to compute the CH-$q$ policy where $q$ is given.
The procedure for how the CH-$q$ defense policy is computed is summarized in Algorithm~\ref{al:chq_policies}. In the following subsection, we discuss the convergence of CH-$q$ policies when $q$ increases.

\subsection{Convergence}
In the following, we show a sufficient condition under which an increase in cognitive levels does not alter the policies for the defender and the adversary.
\begin{Proposition} [Convergence]
\label{prop:convergence}
     Consider Algorithm~\ref{al:chq_policies}, if the adversary does not change their policy in two consecutive CHs, i.e., 
    \begin{align}
    \Ac^{\ell+1} \equiv \Ac^{\ell}, \label{convergence_Ak}
    \end{align}
    then, the adversary and the defender do not alter their policies with CHs higher than $\ell$, i.e, $\Ac^p \equiv \Ac^\ell$ and $\Mc^p \equiv \Ac^\ell$ for all $p \geq \ell + 1$.
    \QET
\end{Proposition}
\begin{proof}
    If the condition \eqref{convergence_Ak} holds, the optimization \eqref{def_Mk} remains unchanged in two consecutive CH-$(\ell+1)$ and CH-$(\ell+2)$ for the defender. As a consequence, the defender policy with such two CHs remains unchanged, i.e., $\Mc^{\ell+2} \equiv \Mc^{\ell+1}$. This results in $\Ac^{\ell+2}  \equiv \Ac^{\ell+1}$ since the CH-$(\ell+2)$ and CH-$(\ell+1)$ adversaries react with the same defender policy. Therefore, the policies for the defender and the adversary remain unchanged for higher CHs.
\end{proof}

The condition \eqref{convergence_Ak} can also be used to stop Algorithm~\ref{al:chq_policies} without further computation before $k$ reaches $q$ since the policies remain unchanged.
In the worst-case scenario, the convergence presented in Proposition~\ref{prop:convergence} occurs at CH-$\binom{N}{\alpha}$ where $\binom{N}{\alpha}$ is the number of all the admissible attack sets.

Alternatively, the convergent defender policy can also be found by solving \eqref{def_Mk} against all the admissible attack sets, i.e.,
\begin{align}
    \hspace{-0.16cm}
    \Mc^{\binom{N}{\alpha}} \triangleq \underset{\Mc \in \Vc, |\Mc| \leq \beta}{\arg \min} \big[ c_s(\Mc) + \max_{\Ac \in \Vc, |\Ac| = \alpha} Q(\Mc, \Ac) \big]. \label{def_Mk_max}
\end{align}
In the next section, we run Algorithm~\ref{al:chq_policies} to examine how deep the reasoning level is to reach the convergent policies in numerical examples.


\section{NUMERICAL EXAMPLES}
In this section, we examine the obtained results through a numerical example of a 10-node network (see Fig.~\ref{fig:problem}) with the number of attack nodes $\alpha = 3$ and the maximum number of monitor nodes $\beta = 3$. All the nodes have the same sensor cost of $1$.
We run Algorithm~\ref{al:chq_policies} to compute the CH-$q$ policy and the corresponding objective functions for the adversary \eqref{def_Ak} and defender \eqref{def_Mk} where $q$ is set at $19$. Numerical results are reported in Tab.~\ref{tab:num_result}. As the defender increases their cognitive level, they strategically adjust their policy to reduce the maximum worst-case disruption caused by the adversary. Conversely, the adversary adapts their attack strategy to maximize disruption as their cognitive reasoning deepens. However, once CH reaches 17, the adversary can no longer find a policy that further increases the worst-case disruption. This suggests that the defender, by increasing their cognitive level to 17, has effectively covered the most critical attack scenarios. This argument is further supported by the last row of Tab.~\ref{tab:num_result}, which presents the result of solving \eqref{def_Mk_max} while considering all admissible attack sets. Beyond CH-17, the strategies for both players remain unchanged, indicating convergence.

\begin{table}[!t]
    \centering
    \begin{tabular}{|c|c|c|c|c|}
    \hline
        CH-$k$ & $\Ac_k$ & $Q(\Mc_k,\Ac_k)$ & $\Mc_k$ & $R(\Mc_k, \{\Ac^{ i}\}_{\forall i < k})$  \\ \hline \hline
         0  & \{2,3,6\}  & 817 & \{1,2,3\}  & 322  \\
        1  & \{5,6,9\}  & 408 & \{4,5,8\}  & 16   \\
        2  & \{2,4,10\} & 390 & \{4,7,8\}  & 27   \\
        3  & \{4,5,10\} & 713 & \{1,4,5\}  & 41   \\
        4  & \{2,3,5\}  & 425 & \{3,7,8\}  & 52   \\
        5  & \{4,6,9\}  & 426 & \{4,5,10\} & 126  \\
        6  & \{5,7,9\}  & 673 & \{3,4,9\}  & 165  \\
        7  & \{3,7,10\} & 350 & \{4,9,10\} & 165  \\
        8  & \{3,8,10\} & 541 & \{2,7,8\}  & 166  \\
        9  & \{1,7,10\} & 469 & \{3,5,9\}  & 176  \\
        10 & \{5,6,10\} & 665 & \{5,7,9\}  & 182  \\
        11 & \{4,6,10\} & 322 & \{1,4,9\}  & 198  \\
        12 & \{1,3,10\} & 351 & \{3,7,9\}  & 229  \\
        13 & \{2,7,8\}  & 427 & \{7,9,10\} & 238  \\
        14 & \{2,6,8\}  & 415 & \{7,8,10\} & 240  \\
        15 & \{2,5,9\}  & 392 & \{3,9,10\} & 273  \\
        16 & \{5,6,7\}  & 390 & \{2,3,9\}  & 285  \\
        17 & \{2,4,10\} & 316 & \{4,7,9\}  & 319  \\
        18 & \{2,4,10\} & 316 & \{4,7,9\}  & 319  \\
        19 & \{2,4,10\} & 316 & \{4,7,9\}  & 319  \\
        \ldots & \ldots & \ldots & \ldots & \ldots \\ 
        $\binom{10}{3}$ & \{2,4,10\} & 316 & \{4,7,9\}  & 319 \\ \hline 
    \end{tabular}
    \caption{CH-$k$ policies and payoffs for the defender and the adversary.
    The last row is computed by solving \eqref{def_Mk_max}. All the cost values are rounded up to the nearest integers.
    }
    \label{tab:num_result}
\end{table}

The results regarding reasoning outcome presented in Propositions~\ref{prop:adv_mismatch}-\ref{prop:def_mismatch} are illustrated in Figs.~\ref{fig:cog_mismatch_Q}-\ref{fig:cog_mismatch_R}. 
In Fig.~\ref{fig:cog_mismatch_Q}, the defender's cognitive level is fixed while the adversary's cognitive level varies. We observe that the adversary achieves the highest reasoning outcome when their cognitive level matches that of the defender, a phenomenon referred to as \textit{reasoning resonance}. This suggests that when the adversary and defender operate at the same cognitive depth, the adversary is best able to exploit the defender’s strategic reasoning. Conversely, in Fig.~\ref{fig:cog_mismatch_R}, where the adversary's cognitive level is fixed while the defender's cognitive level varies, we see that a \textit{cognitive mismatch} tends to benefit the defender, resulting in a more favorable reasoning outcome. This highlights the defender's advantage in situations where they can outthink the adversary, reinforcing the importance of strategic cognitive depth in security decision-making.

\begin{figure}[!t]
    \centering
    \includegraphics[width=\linewidth]{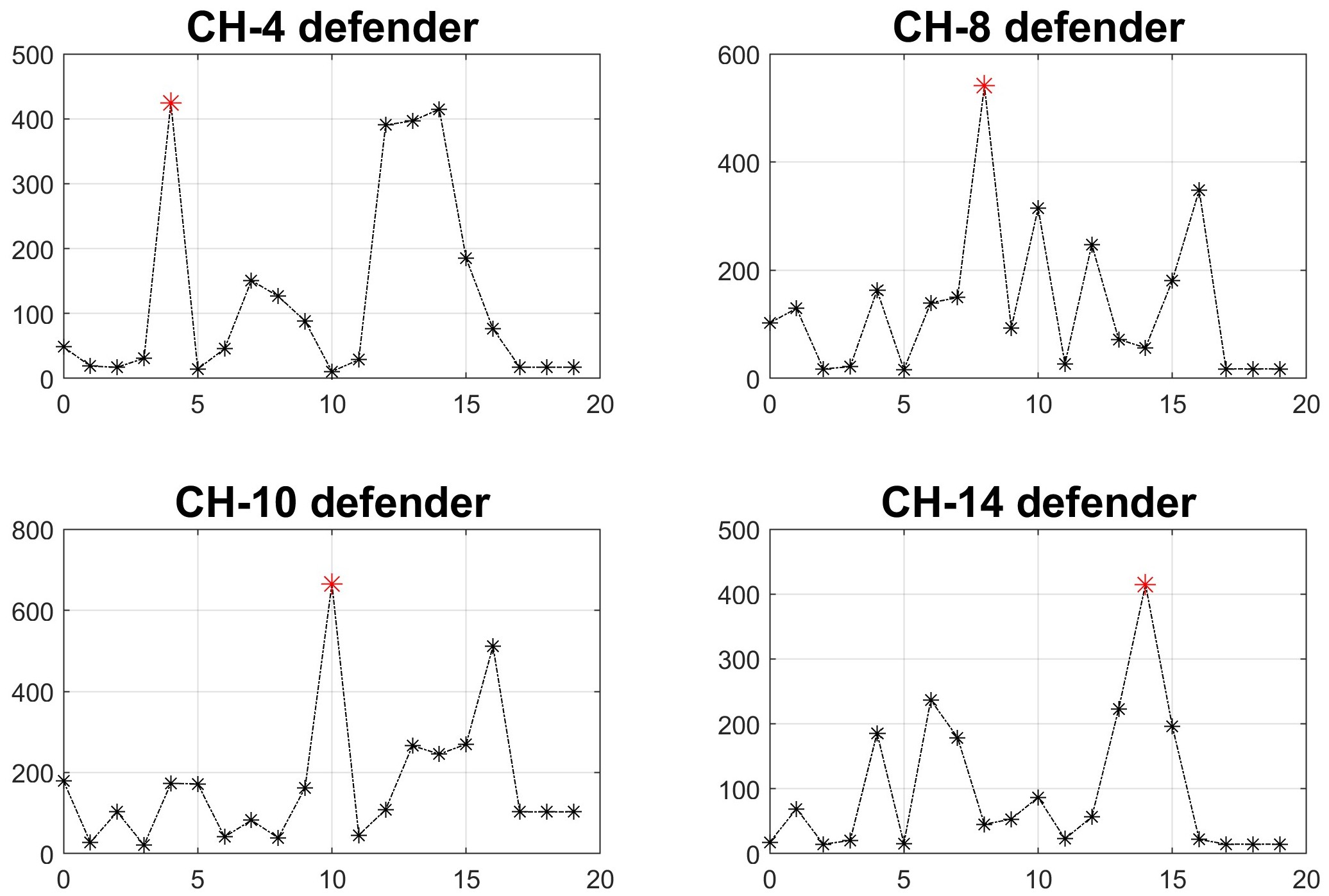}
    \caption{The horizontal axis represents the cognitive reasoning employed by the adversary while the vertical axis indicates the reasoning outcome. The \textit{reasoning resonance}, represented by red asterisks, yields better reasoning outcomes for the adversary.}
    \label{fig:cog_mismatch_Q}
    \includegraphics[width=\linewidth]{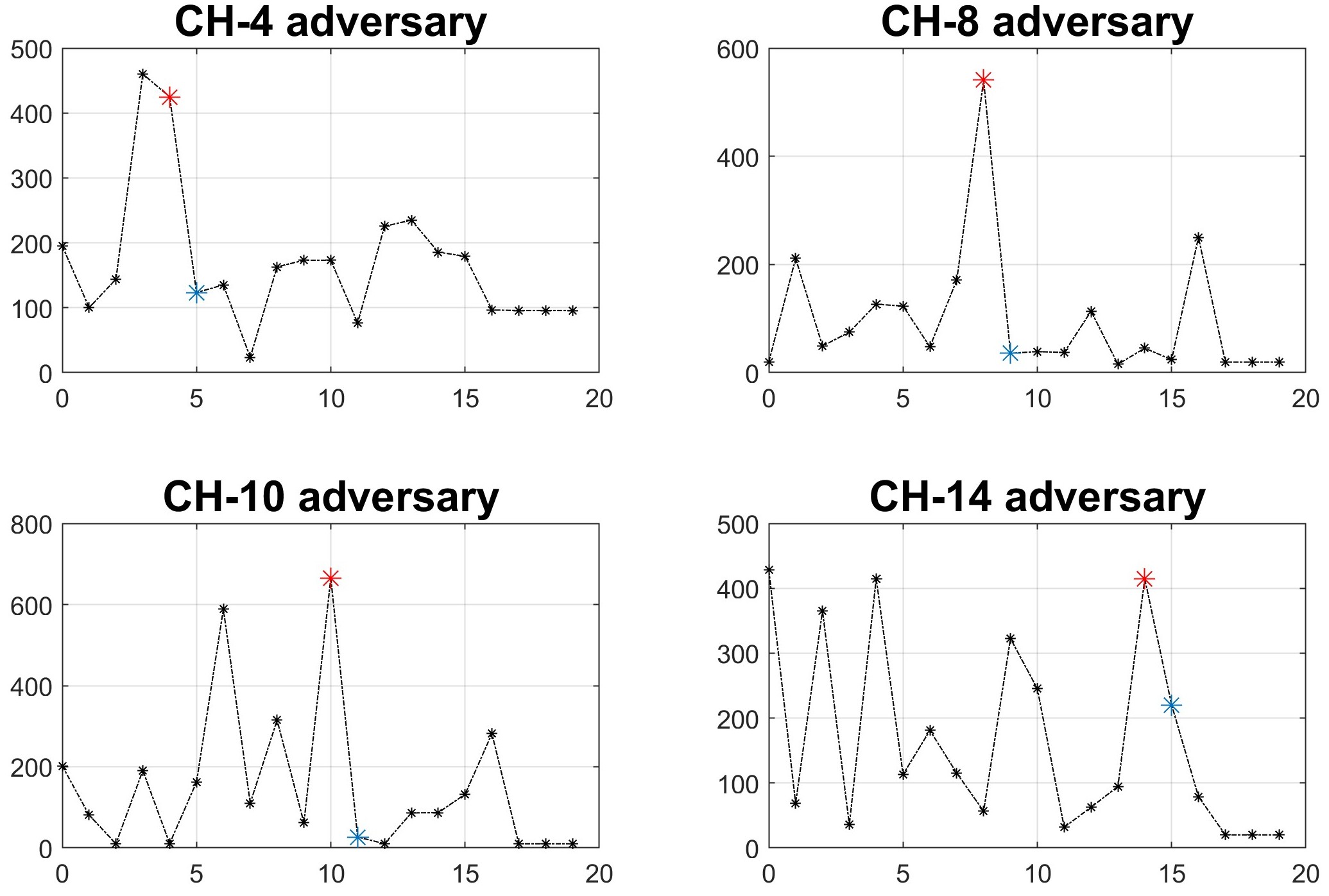}
    \caption{The horizontal axis represents the cognitive reasoning employed by the defender while the vertical axis indicates the reasoning outcome. While red asterisks represent the \textit{reasoning resonance}, blue asterisks indicate that the defender is one step ahead in cognitive reasoning compared to the adversary. The \textit{reasoning mismatch} yields better reasoning outcomes for the defender.}
    \label{fig:cog_mismatch_R}
\end{figure}


\section{CONCLUSIONS}
In this paper, we studied a bilateral cognitive security game in networked control systems, where a strategic adversary launches stealthy false data injection attacks while a defender strategically monitors the system to mitigate disruptions. By integrating cognitive hierarchy models with Stackelberg prediction games, we developed an SDP-based framework to compute optimal attack and defense policies under finite-depth reasoning. We established a convergence condition, showing that increasing cognitive levels beyond a threshold does not alter the players’ strategies. Numerical simulations validated our findings, demonstrating that defenders can achieve optimal security policies while considering only a limited set of attack scenarios.

\section*{APPENDIX}
\subsection{Proof of Lemma~\ref{lem:adversary_chk}}
\label{app:pflem_adversary_chk}
Recall \eqref{def_Ak}, the maximum worst-case disruption $Q(\Mc_k, \Ac_k)$ is equal to $Q_k$. As a consequence, for any other attack set $\Ac$, one has
\begin{align}
    Q(\Mc_k, \Ac) \leq Q(\Mc_k, \Ac_k) = Q_k~~ \forall \Ac \in \Vc, |\Ac| = \alpha. \label{pf:Qk_upper}
\end{align}
By using the optimal variable $\epsilon_\Ac \in \Rbb_{\geq 0}$ for the corresponding attack set $\Ac$, \eqref{pf:Qk_upper} can be rewritten as follows:
\begin{align}
    Q(\Mc_k, \Ac) + \epsilon_\Ac = Q_k~~ \forall \Ac \in \Vc, |\Ac| = \alpha. \label{pf:Qk_equal}
\end{align}
From \eqref{pf:Qk_equal}, one obtain $\epsilon_{\Ac_k} = Q_k - Q(\Mc_k, \Ac_k) = 0$ by definition of $Q_k$. Therefore, finding $\epsilon_\Ac$ and $Q_k$ is equivalent to solving the following optimization problem
\begin{align}
    \max_{\{\epsilon_\Ac\}_{\forall \Ac} \in \Rbb_{\geq 0}} &~~\sum_{\forall \Ac} \epsilon_\Ac \label{opt_epsilon}\\
    \text{s.t.}~~~~& 
        Q(\Mc_k, \Ac) + \epsilon_\Ac \leq Q_k^\star, ~ \forall \, \Ac, \non 
\end{align}
where
\begin{align}
    Q_k^\star = \min_{Q_k \in \Rbb_{>0}}& ~~ Q_k \non  \\ 
    \text{s.t.}~~&~~Q(\Mc_k, \Ac) \leq Q_k, ~ \forall \, \Ac. \non
\end{align}
Solving the two optimization problems in \eqref{opt_epsilon} is equivalent to solving the following optimization problem:
\begin{align}
    \min_{Q_k \in \Rbb_{>0}, \{\epsilon_\Ac\}_{\forall \Ac} \in \in \Rbb_{\geq 0}} &~~ r Q_k - \sum_{\forall \Ac} \epsilon_\Ac 
    \label{Qstar} \\
    \text{s.t.}~~~~~~~~~& Q(\Mc_k, \Ac) + \epsilon_\Ac \leq Q_k, ~ \forall \, \Ac. \non
\end{align}
Here $r$ is chosen as a very large positive number to emphasize the minimization on $Q_k$ and force at least one variable $\epsilon_\Ac$ to be zero simultaneously. The other non-zero values $\epsilon_\Ac$ show the gap between its corresponding worst-case disruption $Q(\Mc_k,\Ac)$ and $Q(\Mc_k,\Ac_k)$.

Next, we show how to compute \eqref{Q_sup} for each given pair of $\Mc_k$ and $\Ac$. The computation is adapted from our previous work \cite[Lemma 3]{nguyen2024Submit} and is reported in the following for a better flow. 
the worst-case disruption \eqref{Q_sup} has the dual form:
    \begin{align}
        &\inf_{ 
        \gamma_\Ac \,\in\, \Rbb^N_{>0} ,\, \psi_\Ac \,\in\, \Rbb^{\alpha}_{>0} } ~ \bigg[ ~ \sup_{ \zeta }~ \bigg\{
         \sum_{m \, \in \, \Mc_k} e_m^\top \gamma_\Ac \big( \delta_m - \norm{y_m^a}_{\Lc_2}^2 \big)
        \non \\
        &
        + \norm{p^a}_{\Lc_2}^2 
        + \sum_{j\,=\,1}^{\alpha} 
    e_j^\top
    \psi_\Ac \big( E - \norm{e_j^\top \zeta}_{\Lc_2}^2 \big) \bigg\}\bigg] \label{Q_inf} \\
        & 
        \text{s.t.}~~ \eqref{sys:xa}-\eqref{sys:yma}, \, x^a(0) = 0,\, x^a(\infty) = 0, \non 
    \end{align}
    where $\gamma_\Ac$ and $\psi_\Ac$ are Lagrange multipliers associated with the first and second inequality constraints in \eqref{Q_sup}, respectively.
    The dual form \eqref{Q_inf} is bounded only if 
    \begin{align}
        &\norm{p^a}_{\Lc_2}^2 - 
        \sum_{m  \in \Mc_k} e_m^\top \gamma_\Ac  \norm{y_m^a}_{\Lc_2}^2 - \sum_{j=1}^\alpha e_j^\top \psi_\Ac \norm{e_j^\top \zeta}_{\Lc_2}^2
          \leq  0, \non 
    \end{align}
    which results in the following optimization problem:
    \begin{align}
    Q(\Mc_k,\Ac) = &
    \inf_{\gamma_\Ac,\, \psi_\Ac}  ~~  \delta^\top \gamma_\Ac  + E \, \,\textbf{1}^\top  \psi_\Ac  \label{Q_min} \\ 
    \text{s.t.}~~~& 
    \eqref{sys:xa}-\eqref{sys:yma}, \, x^a(0) = 0, \, x^a(\infty) = 0, \non \\
    &\gamma_\Ac \, \in \, \Rbb^N_{>0},\, \psi_\Ac \, \in \, \Rbb_{>0}^{\alpha}, \non \\
    &  \hspace{-2cm}
    \norm{p^a}_{\Lc_2}^2 - \sum_{m  \in \Mc_k} e_m^\top \gamma_\Ac  \norm{y_m^a}_{\Lc_2}^2 - \sum_{j=1}^\alpha e_j^\top \psi_\Ac \norm{e_j^\top \zeta}_{\Lc_2}^2 \leq 0.
    \non 
    \end{align}
    The strong duality can be proven by utilizing the lossless S-Procedure \cite[Ch. 4]{petersen2000robust}. Recalling the key results in the dissipative system theory for linear systems \cite{trentelman1991dissipation} with a storage function $S(x^a) \triangleq (x^a)^\top P_\Ac x^a$, where $P_\Ac  \in \Sbb^N$, and a supply rate $s(\cdot,\cdot) \triangleq   \sum_{m  \in \Mc_k} e_m^\top \gamma_\Ac  \norm{y_m^a}_{\Lc_2}^2 + \sum_{j=1}^\alpha e_j^\top \psi_\Ac \norm{e_j^\top \zeta}_{\Lc_2}^2 - \norm{p^a}_{\Lc_2}^2$, we observe that the inequality constraint in \eqref{Q_min} is equivalent to the system being dissipative with respect to the supply rate $s(\cdot,\cdot)$. Hence, the inequality constraint in \eqref{Q_min} can be replaced with the equivalent dissipation inequality and 
    the optimization problem \eqref{Q_min} is translated into the following SDP problem 
    \begin{align}
    \hspace{-0.1cm}
        Q(\Mc_k,\Ac) = &
    \min_{\gamma_\Ac ,\, \psi_\Ac, \,P_\Ac}~  \delta^\top \gamma_\Ac  + E \, \,\textbf{1}^\top  \psi_\Ac \label{Q_min_sdp} \\ 
    \text{s.t.}~~~& \gamma_\Ac \, \in \, \Rbb^N_{>0},\, \psi_\Ac \, \in \, \Rbb_{>0}^{\alpha}, \, P_\Ac \in \Sbb^N, \non \\
    &
    \ba{cc}
    A_c^\top P_\Ac + P_\Ac A_c + W^2 & P_\Ac B_\Ac \\
    B_\Ac^\top P_\Ac & -\, \textbf{diag}(\psi_\Ac)
    \ea \non \\
    &
    - \textbf{diag} \Bigg(\ba{c} 
    \gamma_\Ac \circ z_{k} \\ 0 \ea \Bigg) \preceq 0. \non
    \end{align}
    Substituting \eqref{Q_min_sdp} into \eqref{Qstar} yields \eqref{Q_sdp}, which concludes the proof. \QEDB

\subsection{Proof of Theorem~\ref{thm:chk_defender}}
\label{sec:pfthm_defender_chk}
The proof is adapted from our previous work \cite[Theorem 1]{nguyen2024Submit} and is reported in the following for better reading. Since \eqref{def_Mk} considers the maximum worst-case disruption caused by the adversary with lower CHs, we can leverage the computation of \eqref{def_Ak} by replacing the admissible attack sets with the CH-$i$ adversary policies for all $i \in [0, k-1]$. As a consequence, the optimization problem \eqref{def_Mk} is equivalent to the following optimization:
\begin{align}
    \min_{\Mc_k \in \Vc, |\Mc_k| \leq \beta}&~~~ c_s(\Mc_k) + Q_k, \label{pf:R_min_upper} \\
    \text{s.t.}~~&~~~ Q(\Mc_k, \Ac_i) \leq Q_k ~~ \forall i \in [0, k-1]. \non
\end{align}

Let us recall the policy $\Mc_k$ denoted by $z_k \in \{0,\,1\}^N$ in \eqref{def_zk}.
On the other hand, the sensor budget \eqref{sensor_budget} and the cost of utilized sensors \eqref{cost_sensor} imply the following two constraints:
\begin{align}
    |\Mc_k| &= \textbf{1}^\top z_k \leq \beta, \non \\
    c_s(\Mc_k) &= \sum_{m \, \in \, \Mc_k} \kappa^\top e_m = \kappa^\top z_k. \label{zM_cost}
\end{align}

Next, the computation of $Q(\Mc_k, \Ac_i)$ is directly adopted from \eqref{Q_min_sdp} where the term $\gamma_i \circ z_k$ in the inequality constraint is replaced with the following constraints 
\begin{align}
    \omega_i = \gamma_i \circ z_k,~
    \omega_i \leq M_\infty z_k ~~ \forall i \in [0, k-1], \label{bigM}
\end{align}
where $M_\infty$ is a large positive number, also called ``big M''.

Finally, by substituting \eqref{Q_min_sdp}, \eqref{zM_cost}, and \eqref{bigM} into \eqref{pf:R_min_upper}, we obtain the mixed-integer SDP problem \eqref{R_sdp}. Moreover, since the optimization problem \eqref{Q_sdp} always admits a finite solution and $\kappa^\top z \leq \beta \max(\kappa) < \infty$, the mixed-integer SDP problem \eqref{R_sdp} always admits a finite solution. \QEDB

\bibliographystyle{IEEEtran}
\bibliography{mybibfile}

\end{document}